%% file: main.tex
\documentclass[letterpaper, 10 pt, conference]{ieeeconf}
\renewcommand{\baselinestretch}{0.85}
\IEEEoverridecommandlockouts
\usepackage{cite}
\usepackage{amsmath,amssymb,amsfonts,amsthm}
\usepackage{mathtools}
\usepackage{mathrsfs}
\usepackage{algorithmic}
\usepackage{graphicx}
\usepackage{subcaption}
\usepackage{textcomp}
\usepackage{verbatim}
\usepackage{arydshln}
\usepackage{cuted}
\usepackage{bm}
\usepackage{url}
\usepackage[bookmarks=false]{hyperref}
\usepackage{algorithm}
\usepackage{algorithmic}
\allowdisplaybreaks
\usepackage[dvipsnames]{xcolor}

\usepackage{tikz-cd}
\usepackage{tikz}
\usetikzlibrary{arrows.meta, positioning, calc}

\newtheorem{proposition}{Proposition}
\newtheorem{assumption}{Assumption}

\newtheorem{theorem}{Theorem}
\newtheorem{lemma}{Lemma}

\newtheorem{problem}{Problem}
\newtheorem{definition}{Definition}

\newcommand{\ang}[1]{\left\langle #1\right\rangle}
\newcommand{\matrices}[1]{\begin{bmatrix} #1\end{bmatrix}}

\newcommand{\R}{\mathbb{R}}
\newcommand{\C}{\mathbb{C}}

\newcommand{\rank}{\mathrm{rank}\, }
\newcommand{\kernel}{\mathrm{Ker}\, }
\newcommand{\image}{\mathrm{Im}\, }
\newcommand{\sign}{\mathrm{sign}}
\newcommand{\B}{\mathscr{B}}
\newcommand{\X}{\mathscr{X}}
\newcommand{\Y}{\mathscr{Y}}
\newcommand{\U}{\mathscr{U}}
\newcommand{\V}{\mathscr{V}}

\newcommand{\Ss}{\mathscr{S}}

\newcommand{\W}{\mathscr{W}}
\newcommand{\M}{\mathscr{M}}
\newcommand{\K}{\mathscr{K}}

\newcommand{\Lap}{\mathcal{L}}
\newcommand{\diag}{\mathrm{diag}}
\newcommand{\col}{\mathrm{col}}
\newcommand{\Pwgi}{P_{W_{g,i}^*}}
\newcommand{\Pwi}{P_{W_{i}^*}}
\newcommand{\Wgi}{W_{g,i}^*}

\def\BibTeX{{\rm B\kern-.05em{\sc i\kern-.025em b}\kern-.08em
    T\kern-.1667em\lower.7ex\hbox{E}\kern-.125emX}}

\usepackage{fancyhdr}

\fancypagestyle{myheader}{
    \fancyhf{} 
    \fancyhead[L]{ } 
    \fancyhead[C]{\footnotesize This work has been submitted to the IEEE Control and Decision Conference 2025. Copyright may be transferred without notice, after which this version might no longer be accessible.} 
    \fancyhead[R]{ } 
    \setlength{\headheight}{0pt} 
    \setlength{\headsep}{25pt}    
}

\begin{document}
\title{Bridging Centralized and Distributed Frameworks in Unknown Input Observer Design}

\author{Ruixuan Zhao, Guitao Yang, Peng Li, Boli Chen
\thanks{R. Zhao and B. Chen are with the Department of Electronic and Electrical Engineering, University College London, London, UK \tt\small(ruixuan.zhao.22@ucl.ac.uk; boli.chen@ucl.ac.uk).}%
\thanks{G. Yang is with the Department of Electrical and Electronic Engineering, Imperial College London, London, UK \tt\small(guitao.yang@imperial.ac.uk).}
\thanks{P. Li is with the School of Mechanical Engineering and Automation at Harbin Institute of Technology, Shenzhen, China \tt\small(lipeng2020@hit.edu.cn).}%
}

\pagestyle{myheader}
\begin{titlepage}
    \centering
    \vspace*{\fill}
    {\Huge \bfseries Copyright Statement \\[0.5cm]}

    {\large This work has been submitted to the IEEE Control and Decision Conference 2025. Copyright may be transferred without notice, after which this version might no longer be accessible. \\[2cm]}
    \vspace*{\fill}
\end{titlepage}

\clearpage

\maketitle

\begin{abstract}
State estimation for linear time-invariant systems with unknown inputs is a fundamental problem in various research domains. In this article, we establish conditions for the design of unknown input observers (UIOs) from a geometric approach perspective. Specifically, we derive a necessary and sufficient geometric condition for the existence of a centralized UIO. Compared to existing results, our condition offers a more general design framework, allowing designers the flexibility to estimate partial information of the system state. Furthermore, we extend the centralized UIO design to distributed settings. In contrast to existing distributed UIO approaches, which require each local node to satisfy the rank condition regarding the unknown input and output matrices, our method accommodates cases where a subset of nodes does not meet this requirement. This relaxation significantly broadens the range of practical applications. Simulation results are provided to demonstrate the effectiveness of the proposed design.
\end{abstract}

\section{Introduction}
\label{sec:introduction}

The design of unknown input observers (UIOs) for centralized linear time-invariant systems has been extensively studied in the literature \cite{chen1996design,darouach1994full,valcher1999state}, where necessity and sufficiency of the rank condition regarding the unknown input and output matrices have been established. 
More recently, \cite{bejarano2009unknown} has examined the necessary and sufficient conditions for the existence of a UIO by leveraging the concept of strong detectability, which is detailed in \cite{trentelman2002control}.
Alongside these theoretical advancements, the rapid progress in data-driven methodologies within the control community has spurred the development of new frameworks. For instance, recent studies \cite{turan2021data,disaro2024data} have introduced end-to-end data-driven approaches for UIO design. Notably, in \cite{disaro2024equivalence}, the data-driven conditions presented in \cite{turan2021data,disaro2024data} have been shown to be equivalent to those derived from traditional model-based approaches, as established in \cite{chen1996design,darouach1994full,valcher1999state}. 

Recently, the widespread deployment of embedded systems has empowered sensing devices with integrated communication and computation capabilities, enabling the execution of advanced algorithms directly at the sensor level. This development is particularly beneficial for large-scale systems consisting of multiple interconnected components, where the state space is either high-dimensional or spatially distributed. In response to these challenges, several studies \cite{kim2019completely,han2018simple,wang2024split,yang2022sensor,yang2023plug,zhao2024state,ge2022fixed} have investigated the classical distributed state estimation problem under the assumption that the full input signal is available at each local node. Extending the centralized UIO framework \cite{chen1996design,darouach1994full,valcher1999state}, recent works \cite{yang2022state,cao2023distributed,cao2023distributedauto} have developed distributed UIO schemes to address distributed state estimation with local unknown input signals. These approaches, however, require that each local node individually satisfies the rank condition regarding the unknown input and output matrices, which significantly limits their practical applicability.

The novelty of this paper lies in: 1) Based on the geometric approach, we propose a subspace decomposition technique that leads to a novel centralized UIO design associated with the necessary and sufficient condition; and 2)  
Based on our centralized UIO methodology, a novel distributed UIO framework is proposed. Compared with existing work \cite{yang2022state,cao2023distributed,cao2023distributedauto}, our method does not require each local node to meet the rank condition with respect to unknown input and output matrices, which unlocks wider practical applications.

The structure of this article is organized as follows: Section~\ref{sec:preliminaries} introduces the notations, fundamental concepts, and definitions employed throughout the paper. The problem formulation is presented in Section~\ref{sec:problem statement}. Section~\ref{sec:uio} addresses the state estimation problem in the presence of unknown inputs for both centralized and distributed UIOs. Simulation results validating the proposed approaches are provided in Section~\ref{sec:simu}, and concluding remarks are offered in Section~\ref{sec:conclusion}.

\section{Preliminaries}
\label{sec:preliminaries}
\subsection{Notation}
Let $\mathbb{R}$ and $\mathbb{C}$ denote the sets of real and complex numbers, respectively, and let $\mathbb{R}_{>0}$ represent the set of positive real numbers. A symmetric partition of $\mathbb{C}$ is denoted as $\mathbb{C} = \mathbb{C}_g \cup \mathbb{C}_b$ with $\mathbb{C}_g \cap \mathbb{C}_b = \varnothing$, where $\mathbb{C}_g$ contains the ``good'' (e.g., stable) eigenvalues and $\mathbb{C}_b$ contains the ``bad'' (e.g., unstable) eigenvalues. The identity matrix of dimension $n$ is denoted by $I_n$, and $\mathbf{0}$ denotes a zero matrix of appropriate dimensions. The kernel (null space) and image (column space) of the map $A$ are denoted by $\kernel A$ and $\image A$, respectively. The symbols $\|\cdot\|_1$, $\|\cdot\|_2$, and $\|\cdot\|_\infty$ represent the 1-norm, 2-norm, and infinity norm of a vector or matrix, respectively. The Kronecker product is denoted by $\otimes$, and $\uplus$ indicates the union of sets where repeated elements are retained. The sign function is denoted by $\sign(\cdot)$. The notation $\mathrm{col}(M_1, M_2, \ldots, M_n)$ denotes the vertically stacked matrix $[M_1^\top, M_2^\top, \ldots, M_n^\top]^\top$, while ${\rm diag}(M_1, M_2, \ldots, M_n)$ represents the block-diagonal matrix formed from the matrices $M_i$. The pseudoinverse of a matrix $M$ is denoted by $M^\dagger$. The spectrum of a matrix $M$ is denoted by $\kappa(M)$, and $\sigma_{\min}(M)$ denotes the minimum singular value of $M$.

\subsection{Geometric Approach}
\subsubsection{Basic Definitions}
Let \( A: \mathscr{X} \rightarrow \mathscr{X} \) be an endomorphism, and let \( \mathscr{W} \subseteq \mathscr{X} \) be a subspace with insertion map \( W: \mathscr{W} \rightarrow \mathscr{X} \), such that \( \mathscr{W} = \image W \) and \( W \) are monic. A subspace \( \mathscr{W} \subseteq \mathscr{X} \) is said to be {\em invariant} under \( A \) if \( A\mathscr{W} \subseteq \mathscr{W} \). For an invariant subspace \( \mathscr{W} \), the {\em restriction} of \( A \) to \( \mathscr{W} \) is denoted by \( A|{\mathscr{W}}: \mathscr{W} \to \mathscr{W} \).
The set \( \mathscr{W}_x = x + \mathscr{W} \), for \( x \in \mathscr{X} \), is called a {\em coset} of \( \mathscr{W} \) in \( \mathscr{X} \), and \( x \) is referred to as its representative. The set of all such cosets is denoted by the {\em quotient space} \( \mathscr{X}/\mathscr{W} \coloneqq \{x + \mathscr{W} : x \in \mathscr{X}\} \), which can also be written as \( \frac{\mathscr{X}}{\mathscr{W}} \). The {\em induced map} on the quotient space, denoted \( A|{\mathscr{X}/\mathscr{W}} \), satisfies \( (A|{\mathscr{X}/\mathscr{W}}) P = P A \), where \( P: \mathscr{X} \to \mathscr{X}/\mathscr{W} \) is the canonical projection.
Given a linear map \( C: \mathscr{X} \to \mathscr{Y} \) and a subspace \( \mathscr{S} \subseteq \mathscr{Y} \), the {\em inverse image} of \( \mathscr{S} \) under \( C \) is defined as
\(
C^{-1} \mathscr{S} \coloneqq \{x \in \mathscr{X} \mid Cx \in \mathscr{S} \} \subseteq \mathscr{X}.
\)
For subspaces \( \mathscr{R}, \mathscr{S} \subseteq \mathscr{X} \), their sum and intersection are defined as
\(
\mathscr{R} + \mathscr{S} \coloneqq \{r + s \mid r \in \mathscr{R}, s \in \mathscr{S} \},\ \text{and} \ \mathscr{R} \cap \mathscr{S} \coloneqq \{x \in \mathscr{X} \mid x \in \mathscr{R} \text{ and } x \in \mathscr{S} \}.
\)
The notation \( \mathscr{R} \oplus \mathscr{S} \) indicates that the subspaces $\mathscr{R}$ and $\mathscr{S}$ being added are independent. The orthogonal complement of a subspace \( \mathscr{V} \) is denoted by \( \mathscr{V}^\perp \), and the isomorphism between vector spaces \( \mathscr{V} \) and \( \mathscr{W} \) is indicated by \( \mathscr{V} \simeq \mathscr{W} \). \( \mathrm{Mat}(A) \) denotes the matrix representation of the map \( A \).

\subsubsection{$(C,A)$-invariant Subspace}
Let $A:\X\rightarrow \X$ and $C:\X\rightarrow \Y$. A subspace $\W \subseteq \X$ is said to be $(C,A)$-invariant if there exists a map $L:\Y\rightarrow \X$ such that
\begin{equation}\label{eq:def(C,A)-inv}
    (A+LC)\W \subseteq \W.
\end{equation}
We denote by $\mathbf{L}(\W)$ the set of all maps $L$ satisfying \eqref{eq:def(C,A)-inv}, assuming $A$ and $C$ are clear in the context.
Given a subspace $\B\subseteq \X$, we define $\underline{\W}(C,A;\B)$ as the family of $(C,A)$-invariant subspaces containing $\B$, which
has an infimal element denoted by $\W^*(C,A;\B)$ or simply $\W^*(\B)$ if $A$ and $C$ are clear from context.
\subsubsection{Unobervability Subspace}
A subspace $\Ss \subseteq \X$ is said to be an unobservability subspace if there exists $L:\Y \rightarrow \X$ (output injection) and $H:\Y \rightarrow \Y$ (measurement mixing) such that
\begin{equation}\label{eq:def_Ss}
    \Ss = \ang{\kernel HC\,|\,A+LC }.
\end{equation}
 Note that the \textit{unobservability subspace} is a different concept from the \textit{unobservable subspace}\footnote{The unobservable subspace is defined as $\ang{\K\,|\, A}\coloneqq \K \cap  A^{-1} \K \cap A^{-2} \K \cap \cdots \cap A^{-n+1} \K$ with $\K = \kernel C$.}.

We denote by $\mathbf{L}(\Ss)$ the set of all maps $L$ that satisfy \eqref{eq:def_Ss}, assuming $A$ and $C$ are clear in the context.
Given a subspace $\B\subseteq \X$, the family of unobservability subspaces containing $\B$ is denoted by $\underline{\Ss}(C,A;\B)$, which also has an infimal element denoted by $\Ss^*(C,A;\B)$ or simply $\Ss^*(\B)$ when the context is clear.

\section{Problem Statement}
\label{sec:problem statement}
Consider the following linear time-invariant (LTI) system
\begin{equation}\label{eq:sys}
\left\{
\begin{aligned}
        &\dot x (t) = Ax(t) + Bu(t) \, ,\\
        &y(t) = C x(t)\, ,
\end{aligned}\right.
\end{equation}
where ${A\in \mathbb{R}^{n\times n}}$, ${B\in \mathbb{R}^{n\times m}}$ and $C\in\mathbb{R}^{p\times n}$ are known system matrices. ${x \in \mathbb{R}^n}$ is the state vector, ${u\in \mathbb{R}^m}$ is the control input, and $y\in\mathbb{R}^p$ is the output measurement.
\begin{problem}\label{pro:uio}
Partition system's input signal into
    \begin{equation}
        Bu = \Acute{B}\Acute{u} + \Bar{B}\Bar{u}
    \end{equation}
with $\Acute{B} \in \mathbb{R}^{n\times l}$, $\bar{B} \in \mathbb{R}^{n\times (m - l)}$, where $m-l \leq p \leq n$. $\acute u \in \mathbb{R}^{l}$, and $\bar u \in \mathbb{R}^{m - l}$ are the known and unknown input signals, respectively. 

The objective is to design a UIO to estimate the state vector $x$ of \eqref{eq:sys} while having only access to the partially known input signal $\acute u$, and measurement output $y$.
\end{problem}
\begin{definition}
Let $\hat{x}$ be the estimate of $x$ produced by the state observer $\mathcal{O}$. $\mathcal{O}$ is a centralized UIO for system \eqref{eq:sys} associated with {\em Problem~\ref{pro:uio}} if 
$
    \lim_{t\rightarrow\infty}\|x-\hat x\|=0.
$
\end{definition}

\begin{problem}\label{pro:duio}
    For a large-scale system \cite{kim2019completely,han2018simple,wang2024split,yang2022sensor,yang2023plug}, the system output could be sensed by a group of sensors
    \begin{equation}
        y_i=C_i x,\ i\in\mathbf{N}
    \end{equation}
    with $y=\col (y_1,y_2,\cdots,y_N)$, where $y_i\in\mathbb{R}^{p_i}$, $\sum_{i=1}^N p_i=p$ and $C=\col (C_1,C_2,\cdots,C_N)$.
The sensors communicate via a network represented by an undirected graph denoted by $\mathcal{G} = (\mathbf{N},\mathcal{E},\mathcal{A})$, where $\mathbf{N}= \{1,2,\dots,N\}$ is a finite nonempty set of nodes of the graph (describing the networked observer containing $N$ local sensors), $\mathcal{E}\subseteq \mathbf{N} \times \mathbf{N}$ represents the edges of the graph (describing communication among the nodes), and $\mathcal{A}=[a_{ij}] \in \mathbb{R}^{N\times N}$ is the adjacency matrix, where $a_{ij}=a_{ji}=1$ if there exists an edge between node $i$ and node $j$, and $a_{ij}=a_{ji}=0$ otherwise.
$\mathcal{L}$ is defined as the Laplacian matrix associated with graph $\mathcal{G}$.
    
At each node $i$, the input term of the system can be partitioned by
    \begin{equation}\label{eq:inputdecomp}
    Bu = B_i u_i +\bar{B}_i \bar{u}_i,
    \end{equation}
with $B_i \in \mathbb{R}^{n\times l_i}$, $\bar{B}_i \in \mathbb{R}^{n\times (m - l_i)}$, where $m-l_i \leq p_i \leq n$. $u_i \in \mathbb{R}^{l_i}$ and $\bar u_i \in \mathbb{R}^{m - l_i}$ are the known and unknown input signals for node $i$, respectively. 

The objective is to design a Distributed UIO $\{\mathcal{O}_i\}_{i\in\mathbf{N}}$ to reconstruct the state vector $x$ at each node, while each node $i$, $i\in\mathbf{N}$, has only access to its local known control input $u_i$ and local measurement $y_i$.
\end{problem}
\begin{definition}
Let $\hat{x}_i$ be the estimate of $x$ produced by local observer $\mathcal{O}_i$. $\{\mathcal{O}_i\}_{i\in\mathbf{N}}$ is a distributed UIO for system \eqref{eq:sys} associated with {\em Problem~\ref{pro:duio}} if for all $i\in\mathbf{N}$, 
$
    \lim_{t\rightarrow\infty}\|x-\hat x_i\|=0.
    $
\end{definition}

\section{Unknown Input Observer}\label{sec:uio}
To introduce our novel UIO design, we first introduce a subspace decomposition based on the geometric approach.
\subsection{$\W_g^*$ Subspace Decomposition}\label{sec:space_decomposition}
To estimate the most information of system states in the sense of subspace, we aim to identify the infimal $(C, A)$-invariant subspace $\W_{g}^*$ that contains $\image \bar{B}$, while simultaneously allowing the assignment of the spectrum of the induced map $A_L|\X/\W_{g}^*$\footnote{To simplify the notation, we define $A_{L} \coloneqq A+LC$, where $L: \Y \to \X$ is the output injection map.} to lie entirely within the desirable region of the complex plane—referred to as the ``good'' part $\C_g$—through an appropriate choice of the output injection map $L$. This construction allows for the estimation of $P_{W_{g}^*} x$ without interference from the unknown input $\bar u$, where $P_{W_{g}^*} : \X \to \X/\W_{g}^*$ is the canonical projection. Moreover, this projection is designed to minimize information loss, in the sense that the dimension of $\W_g^*=\kernel P_{W_{g}^*}$ is as small as possible.

Before identifying $\W_g^*$, we begin with the analysis of the invariant zeros of the system, which is associated with the quotient space $\Ss^*/\W^*$.
Let $\beta(\lambda)$ denote the minimal polynomial of $A_L|\Ss^{*}/\W^*$. We factorize $\beta(\lambda)$ as $\beta(\lambda) = \beta_{g}(\lambda) \beta_{b}(\lambda)$, where the zeros of $\beta_{g}(\lambda)$ in $\C$ lie within $\C_g$, while those of $\beta_{b}(\lambda)$ lie within $\C_b$, and write
%
    $\bar{\X}_{g}^* \coloneqq \frac{\Ss^{*}}{\W^{*}}\  \bigcap\  \kernel \beta_{g}(A_L|\Ss^*/\W^*),\ 
    \bar{\X}_{b}^* \coloneqq \frac{\Ss^{*}}{\W^{*}}\  \bigcap \ \kernel \beta_{b}(A_L|\Ss^*/\W^*).$
%
Then, we can obtain
\begin{equation}\label{eq:Xa_Xb}
    \frac{\Ss^{*}}{\W^{*}} = \bar{\X}_{g}^*\oplus \bar{\X}_{b}^*.
\end{equation}
This decomposition effectively ``splits'' $\Ss^*/\W^*$ associated with invariant zeros into two sub-quotient subspaces, each corresponding to the ``good'' and ``bad'' modes, respectively. We now introduce the following lemma to determine $\W_{g}^*$.
%
\begin{lemma}\label{thm:Wg}
    \cite{zhao2025DUIO} Let $P_{W^*}:\X \to \X/\W^*$ be the canonical projection, where $\W^*$ is the infimal $(C,A)$-invariant subspace.
    Then, the subspace $\W_{g}^*$ defined as
    \begin{equation}\label{eq:Wg}
        \W_{g}^*
        \coloneqq 
        P_{W^*}^{-1}\bar{\X}_{b}^* 
    \end{equation}
    is the infimal $(C, A)$-invariant subspace containing $\image \bar{B}$, while enabling the assignment of $\kappa(A_L|\X/\W_{g}^*)$ into the partial complex plane $\mathbb{C}_g$.
\end{lemma}
%
\begin{figure}[htp] 
\centering
\scalebox{0.9}{
    \input{Figures/Wg_decompose.tikz}}\\[-1.5ex]
    \caption{Commutative diagram of $\W_{g}^*$ decomposition.}
    \label{fig:Wg_decompose}
\end{figure}
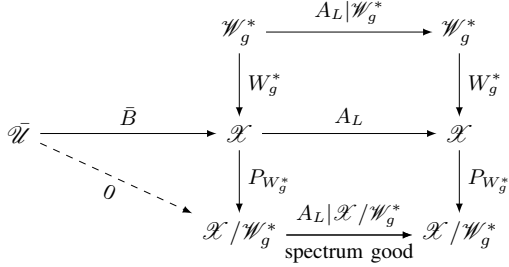
With the desired subspace $\W_{g}^*$ identified, we can now efficiently reconstruct $P_{W_{g}^*}x$ at each node. Notably, this reconstruction remains unaffected by the unknown input $\bar{u}$, as illustrated by the commutative diagram in Fig.~\ref{fig:Wg_decompose}. This subspace decomposition mechanism serves as the foundation for our UIO design, which we introduce in the next subsection.

\subsection{Centralized UIO Design}\label{sec:centralized_UIO}
Based on the subspace decomposition in section~\ref{sec:space_decomposition}, we propose the centralized UIO in the following form
    \begin{equation}\label{eq:cen_uio}
    \begin{aligned}
        &\dot z = \bar{\bar A}_{L}z + P_{W_g^*}\acute B\acute u - P_{W_g^*}Ly,\\
        &\hat{x} = E z + F y
    \end{aligned}
    \end{equation}
where $z\in\mathbb{R}^{n-w_g^*}$ is the auxiliary variable, $\bar{\bar A}_{L}=\mathrm{Mat}(A_{L}|\X/\W_{g}^*)$ and $P_{W_g^*}$ are defined in Section~\ref{sec:space_decomposition}, $E\in\mathbb{R}^{n\times(n-w_g^*)}$ and $F\in\mathbb{R}^{n\times p}$ are specified by \eqref{eq:EPFC}.
\begin{lemma}
    For an LTI system \eqref{eq:sys} associated with {\em Problem~\ref{pro:uio}}, \eqref{eq:cen_uio} is a centralized UIO if and only if
    \begin{equation}
        \W_{g}^*\cap \kernel{C}=0. \label{eq:UIO_con}
    \end{equation}
\end{lemma}
\begin{proof}
    (Sufficiency) Suppose that \eqref{eq:UIO_con} holds, we have already established the existence of $L\in\mathbf{L}(\W_{g}^*)$ such that
    $\sigma\left(A_L|\X/\W_g^*\right)\subset\C_g$.
    Since $\bar{\bar A}_{L}=\mathrm{Mat}( A_{L}|\X/\W_{g}^*)$, it follows that $\bar{\bar A}_{L}$ is Hurwitz. Moreover, \eqref{eq:UIO_con} implies that
       $ {\W_g^*}^\bot + \image{C^\top} = \X$. 
    Since ${\W_g^*}^\bot \simeq \image{P_{W_g^*}^\top}$, there must exist matrices $E\in\mathbb{R}^{n\times (n-w_g^*)}$ and $F\in\mathbb{R}^{n\times p}$ such that
    \begin{equation}\label{eq:EPFC}
        E P_{W_g^*}+FC=I_n.
    \end{equation}
    Defining the estimation error as $e:=x-\hat x$, and combining \eqref{eq:EPFC}, we obtain 
        $e=E\left(P_{W_g^*}x-z\right).$
    To establish the convergence of $e$, it suffices to analyze the stability of $\zeta:=P_{W_g^*}x-z$, whose dynamics evolve as
    \begin{equation*}
        \dot \zeta = P_{W_g^*}\left(Ax+\Acute{B}\Acute{u} + \Bar{B}\Bar{u}\right)-\bar{\bar A}_{L}z - P_{W_g^*}\acute B\acute u + P_{W_g^*}Ly.
    \end{equation*}
    Thanks to $P_{W_g^*}\bar B=0$, it simplifies to
    $\dot \zeta = P_{W_g^*}A_Lx-\bar{\bar A}_{L}z$. 
Referring to Fig~\ref{fig:Wg_decompose}, we note that $P_{W_g^*}A_L=\left(A_{L}|\X/\W_{g}^*\right)P_{W_g^*}$, leading to
    \begin{equation}
        \dot \zeta = \bar{\bar A}_{L}\left(P_{W_g^*}x-z\right)=\bar{\bar A}_{L}\zeta.
    \end{equation}
    Since $\bar{\bar A}_{L}$ is Hurwitz, $\zeta$ converges to $0$ asymptotically. Consequently, $e=E
    \zeta$ also converges to $0$ asymptotically.
    
    (Necessity) To prove the necessity of \eqref{eq:UIO_con}, we proceed by contradiction. Assume that there exists a non-trivial subspace $\V$ such that
    \begin{equation}\label{eq:contradict}
        0 \subset \V=\W_{g}^*\cap \kernel{C}
    \end{equation}
    and the system state can be reconstructed. From \eqref{eq:contradict}, it follows that $\V\subseteq\W_g^*$ and $\V\subseteq\kernel{C}$.

    Since $\V\subseteq\W_g^*$, any state vector lying in $\V$ is either influenced by unknown input $\bar u$ or resides in the subspace $\bar{\X}_b^*$, which is associated with unstable invariant zeros. Consequently, any non-zero state vector in $\V$ cannot be estimated asymptotically under the influence of the unknown input through the channel $\bar B$. On the other hand, $\V\subseteq\kernel{C}$ implies that any state vector lying in $\V$ cannot be directly inferred from the output $y$. Overall, if \eqref{eq:contradict} holds, the entire system state cannot be reconstructed, which concludes the necessity of \eqref{eq:UIO_con}.
\end{proof}
%
%
\begin{proposition}\label{prop:equivalence}
     The existing condition for centralized UIO \cite{chen1996design,darouach1994full,valcher1999state},
     specifically {\em i}) $\rank (C\bar B) = \rank (\bar B)$, {\em ii}) $(C,A_1)$ is a detectable pair with $A_1 = \left( I_n - \bar B (C\bar B)^\dagger C\right) A$,
is equivalent to the geometric condition proposed in \eqref{eq:UIO_con}.
\end{proposition}
\begin{proof}
    To establish the proof, we first reformulate the conditions i) and ii) through the lens of geometric approach. Condition i) indicates that $\image \bar{B} \cap \kernel C = 0$. 
    In condition ii), we observe that $\bar B (C\bar B)^\dagger C$ is a projection map on $\image \bar B$. Consequently, $\left( I_n - \bar B (C\bar B)^\dagger C\right)$ is the projection on $\M$ along $\W^*$ \cite[Chap. 0.4]{wonham1985linear}, where $\M \simeq \X/\W^*$ and $\W^*$ is the infimal $(C,A)$-invariant subspace containing $\image{\bar B}$. Accordingly, the condition that $(C,A_1)$ is detectable is equivalent to $(\bar C, \bar A_L)$ ($\bar C$, $\bar A_L$ are defined as in Fig.~\ref{fig:forremark}) being detectable, which means the sub-quotient space associated with unstable invariant zeros is trivial, i.e., $\bar\X^*_b=0$. 
    Thus, conditions i) and ii) are equivalent to $\image \bar{B} \cap \kernel C = 0$, $\bar\X^*_b=0$. With the reformulation, we now proceed the proof.

    $\big($i) \& ii) $\Rightarrow$ \eqref{eq:UIO_con}$\big)$
    According to \cite[Lemma~4]{massoumnia1986geometric}, a subspace $\W$ is $(C,A)$-invariant if and only if $A(\W \cap \kernel C) \subseteq \W$. Given that $\image \bar{B} \cap \kernel C = 0$, it follows that $A(\image \bar B \cap \kernel C) =0 \subseteq \image \bar B$, which confirms that $\image \bar B$ is $(C,A)$-invariant. Moreover, as $\W(\image \bar B) \supseteq \image \bar B$, $\image \bar B$ is evidently the infimal subspace of all $\W(\image \bar B)$. Consequently, we conclude that $\W^*(\image \bar B)=\image \bar B$. From \eqref{eq:Wg}, it follows that $\W_g^* = P_{W_g^*}^{-1}\bar\X^*_b = P_{W_g^*}^{-1} 0 = \W^*=\image \bar B$, which implies that $\W_{g}^*\cap \kernel{C}=0$ holds.
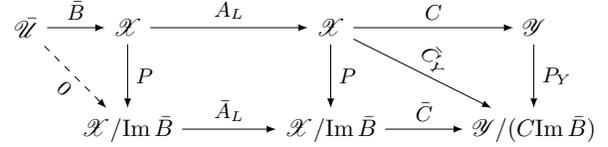
\begin{figure}[t]
    \centering
    \scalebox{0.9}{\input{Figures/forremark.tikz}}\\[-1.2ex]
    \caption{Commutative diagram for systems satisfying conventional UIO conditions. In the diagram, $P:\X\to \X/\image \bar B$ is the canonical projection modulo $\image \bar B$ (note that $\image \bar B$ is now $(C,A)$-invariant), and $P_{Y}:\Y\to\bar \Y$ is the canonical projection modulo $C\image \bar B$. 
    The map $\bar{C}:\X / \image \bar B \to \Y / (C\image \bar B)$ is well-defined due to the fact that $\kernel {\bar{C}_Y}=\kernel{P_{Y}C}=\image \bar B+\kernel C \supseteq \image \bar B$. The map $\bar{A}_{L}:\X / \image B \to \X / \image B$ is the map induced on $\X / \image B$ by $A_{L}$.}
    \label{fig:forremark}
\end{figure}

$\big($\eqref{eq:UIO_con} $\Rightarrow$ i) \& ii)$\big)$ Since $\image \bar B \subseteq\W_g^*$, it is straightforward to conclude that $\image \bar B \cap \kernel C =0$ follows directly from $\W_{g}^*\cap \kernel{C}=0$, which further implies $\W^* = \image{\bar{B}}$ by following the same derivation as in the previous paragraph. According to \cite[Eq. (2.62)]{massoumnia1986geometric} and the fact $\Ss^*=P_{W^*}^{-1}(\Ss^*/\W^*)$, we derive $P_{W^*}^{-1}(\Ss^*/\W^*)+\kernel{C} = \W^* + \kernel{C}$. Furthermore, using \cite[Eq. (2.9)]{massoumnia1986geometric}, we obtain $\Ss^*/\W^* + P_{W^*}\kernel{C} = P_{W^*}\kernel{C}$, which implies that $\Ss^*/\W^*\subseteq P_{W^*}\kernel{C}$. Combining this with \eqref{eq:Xa_Xb}, we write
\begin{equation}
    \bar{\X}^*_b\subseteq P_{W^*}\kernel{C}.\label{eq:(14)}
\end{equation}
Follows from $\W^*_g\cap \kernel{C} = 0$ and the definition in \eqref{eq:Wg}, we conclude that $P_{W^*}^{-1}\bar{\X}^*_b\cap \kernel{C} = 0$. 
Observe that $\kernel P_{W^*}=\W^*\subseteq\W^*_g$, which implies that $\W^*\subseteq(\W^*_g + \kernel C)$. Consequently, we have $\W^*\cap(\W^*_g + \kernel C) = \W^*$. Moreover, since $\W^*\cap\kernel{C}=0$, it follows that $\W^*\cap\W^*_g +\W^*\cap \kernel C = \W^*$.
Then, by referring to \cite[Chapter~0.4, Eq.(4.2), Eq.(4.3)]{wonham1985linear}, we derive
\begin{align}
    P_{W^*}(P_{W^*}^{-1}\bar{\X}^*_b\cap &\kernel{C})\notag \\ 
    &= (P_{W^*} P_{W^*}^{-1}\bar{\X}^*_b)\cap P_{W^*}\kernel{C}\!=\!0.\label{eq:(15)}
\end{align}
From \cite[Eq. (2.27 )]{massoumnia1986geometric}, it follows that $P_{W^*} P_{W^*}^{-1}\bar{\X}^*_b = \bar{\X}^*_b\cap(\X/\W^*)=\bar{\X}^*_b$. Thus, using \eqref{eq:(15)}, we establish
\begin{equation}
    \bar{\X}^*_b\cap P_{W^*}\kernel{C} = 0.\label{eq:(16)}
\end{equation}
Finally, combing \eqref{eq:(14)} and \eqref{eq:(16)}, we conclude that $\bar{\X}^*_b=0$.

Overall, since both directions hold, i.e., i) \& ii) $\Rightarrow$ \eqref{eq:UIO_con} and \eqref{eq:UIO_con} $\Rightarrow$ i) \& ii), the existing UIO condition and \eqref{eq:UIO_con} are shown to be equivalent. This completes the proof.
\end{proof}

\subsection{Distributed UIO Design}
\label{sec:distributed_uio}
 Building upon the methodologies developed in \cite{chen1996design,darouach1994full,valcher1999state}, several notable studies \cite{yang2022state,cao2023distributed,cao2023distributedauto} have explored the distributed UIO problem. However, these approaches impose a stringent requirement: each local node must individually satisfy the rank condition regarding the unknown input and output matrices. This condition may lead to infeasible designs, even if only a single node fails to meet it.

To circumvent the limitation mentioned above, we extend our centralized UIO proposed in Section~\ref{sec:centralized_UIO} to distributed settings. First, based on the geometric approach in Section~\ref{sec:space_decomposition}, for each node $i$, we make some definitions: $\W_{i}^*\!=\!\W_{i}^*(C_i,A;\bar B_i), \Ss_{i}^*\!=\!\Ss_{i}^*(C_i,A;\bar B_i),\W_{g,i}^*\!=\!P_{W_i^*}^{-1}\bar{\X}_{b,i}^*$,
where $P_{W_i^*}:\X\rightarrow\X/\W_i^*$ is the canonical projection, $\bar{\X}_{b,i}^*$ is the sub-quotient space associated with unstable invariant zeros defined by $\bar{\X}_{b,i}^* \coloneqq \frac{\Ss_i^{*}}{\W_i^{*}}\  \bigcap \ \kernel \beta_{b}(A_{L_i}|\Ss_i^*/\W_i^*)$. Let $V_i$ denote the orthonormal basis of $\bar\X^*_{b,i}$ at node $i$, i.e., $\image{V_i}\simeq \bar{\mathscr{X}}_{b,i}^*$. Moreover, we define the insertion maps $W_{i}^*:\W_{i}^*\rightarrow\X$ and $W_{g,i}^*:\W_{g,i}^*\rightarrow\X$, and the canonical projection $\Pwgi:\X\rightarrow\X/\W_{g,i}^*$.

For the node set $\mathbf{N}$, we distinguish it into two subsets $\mathbf{N}_1$ and $\mathbf{N}_2$, i.e., $\mathbf{N}=\mathbf{N}_1\cup\mathbf{N}_2$, where $\mathbf{N}_1$ denotes the set of sensors that satisfy the local rank condition regarding local unknown input and local output matrices \cite{yang2022state,cao2023distributed,cao2023distributedauto}
\begin{equation}
    \rank (C_i\bar B_i)=\rank (\bar B_i),\label{eq:duio_rank_con}
\end{equation}
$\mathbf{N}_2$ denotes the set that does not satisfy the rank condition, i.e., $\rank (C_j\bar B_j)\neq \rank (\bar B_j)$, $j\in\mathbf{N}_2$. Therefore, we need to design two types of local observers.

For the node set $\mathbf{N}_1$, since \eqref{eq:duio_rank_con} holds, we can state that $\W_i^*=\image{\bar B_i}$ and $\W_i^*\cap\kernel{C_i}=0$
based on the proof of Proposition~\ref{prop:equivalence}. Moreover, since $\kernel{P_{W_i^*}}=\W_i^*$, there must exist matrix $E_i\in\mathbb{R}^{n\times(n-w_i^*))} $ and $F_i\in\mathbb{R}^{n\times p_i}$ such that
\begin{equation}\label{eq:F_i}
    E_i P_{W_i^*} + F_i C_i = I_n,\ i\in\mathbf{N}_1.
\end{equation}
The dynamics of the distributed UIO for $\mathbf{N}_1$ are designed as
\begin{align}
        \dot z_i &=\bar{A}_{L_i}z_i - P_{W_i^*}L_iy_i + P_{W_i^*}B_i u_i \notag\\
        &+ \chi P_{W_i^*}V_iV_i^\top \sum_{j=1}^N a_{ij} (\hat x_j - \hat x_i)\notag\\
        \hat x_i &= E_i z_i + F_iy_i,\ i\in\mathbf{N}_1
\label{eq:DUIO_1}
\end{align}
where $z_i\in\mathbb{R}^{n-w_{i}^*}$ is the auxiliary variable, $\chi\in\R_{>0}$ is the coupling gain of the consensus term, 
$\bar{A}_{L_i}:=\mathrm{Mat}(A_{L_i}|\X/\W_i^*)$, $E_i$ and $F_i$ are defined in \eqref{eq:F_i}.

The dynamics of the Distributed UIO for $\mathbf{N}_2$ are designed as follows
\begin{align}
    &\dot{\hat{x}}_i \!=\! A_{L_i}\hat{x}_i-L_i y_i + B_i u_i 
    + \chi W_{g,i}^* {W_{g,i}^*}^\top \sum_{j=1}^N a_{ij} (\hat x_j\!-\!\hat x_i)\notag\\
    &\ + \gamma W_{g,i}^*\mathrm{sign}\left( {W_{g,i}^*}^\top \sum_{j=1}^N a_{ij} (\hat x_j - \hat x_i) \right),\ i\in\mathbf{N}_2,    
\label{eq:DUIO_2}
\end{align}
where $\gamma\in\R_{>0}$ is the coupling gain of the nonlinear consensus terms. 
\begin{assumption}\label{asm:connected}
    The undirected graph $\mathcal{G} = (\mathbf{N},\mathcal{E},\mathcal{A})$ that describes the communication connection among the local distributed UIOs is connected.
\end{assumption}
\begin{assumption}\label{asm:bounded}
    The unknown inputs of the node set $\mathbf{N}_2$ are bounded, i.e., for all $i\in\mathbf{N}_2$,
        $\|\bar{u}_i\|_\infty \le \bar{u}_{\max}$ holds.
\end{assumption}
\begin{assumption}\label{asm:geo_con}
    The sub-quotient space $\bar\X_{b,i}^*$ ($i\in\mathbf{N}_1$) and the subspace $\W_{g,j}^*$ ($j\in\mathbf{N}_2$) have the joint property: $\left(\cap_{i\in\mathbf{N}_1}\image{V_i}\right)\bigcap\left(\cap_{j\in\mathbf{N}_2}\W_{g,j}^*\right)=0$, where $\image{V_i}\simeq \bar{\mathscr{X}}_{b,i}^*$.
\end{assumption}
Building upon the above assumptions, we present the following result.
\begin{theorem}
    Given an LTI system \eqref{eq:sys} and {\em Problem~\ref{pro:duio}}, under {\em Assumptions~\ref{asm:connected}-\ref{asm:geo_con}, the networked observer $\{\mathcal{O}_i\}_{i\in\mathbf{N}}$ obtained by combining \eqref{eq:DUIO_1} and~\eqref{eq:DUIO_2}} is a distributed UIO if
    \begin{equation}
    \begin{split}
        \chi&> \frac{\left\|\bm{A}_L\right\|_2}{\sigma_{\min}\left({\bm W_V}^\top(\Lap\otimes I_n){\bm W_V}\right)},\\ 
        \gamma&>\bar{u}_{\max}\max_{i\in\mathbf{N}_2}\left(\left\| \bar{B}_i \right\|_1\right)\max_{i\in\mathbf{N}_2}\left(\left\|W_{g,i}^*\right\|_{\infty}\right)
    \end{split}\label{con:gain}
    \end{equation}
    holds, where
    \begin{equation}
        \begin{aligned}
            \bm W_V&=\diag\left(\diag_{i\in\mathbf{N}_1}(V_i),\diag_{j\in\mathbf{N}_2}(W_{g,j}^*)\right)\\
    \bm A_L&=\diag\left(\diag_{i\in\mathbf{N}_1}(\Tilde{\Tilde{A}}_{L_i}),\diag_{j\in\mathbf{N}_2}(\Tilde{A}_{L_j})\right)
        \end{aligned}
    \end{equation}
    with $\Tilde{\Tilde{A}}_{L_i}:=\mathrm{Mat}(A_{L_i}|\bar\X_{b,i})$ and $\Tilde{A}_{L_i}:=\mathrm{Mat}(A_{L_i}|\W_{g,i}^*)$.
\end{theorem}
\begin{proof}
    Let $e_i:=x-\hat x_i$ be the estimation error at node $i$.
According to the invariant property, we have the relationships
\begin{equation}
    P_{W_i^*}=\begin{bmatrix}
        P_{W_{g,i}^*}\\V_i^\top
    \end{bmatrix},\ W_{g,i}^*=\begin{bmatrix}
        W_i^*&V_i
    \end{bmatrix}.
\end{equation}
For $i\in\mathbf{N}_1$, we have $e_i=E_i(P_{W_i^*}x-z_i)$. According to \eqref{eq:F_i}, we can always choose a $E_i$ such that $\image{E_i}\simeq {\X/\W_i^*}$, and choose a matrix $M_i\in\mathbb{R}^{n\times n}$ such that $e_i=M_i\varrho_i$ with $\varrho_i:=P_{W_i^*}^\top(P_{W_i^*}x-z_i)$ and $\|M_i\|\leq 1$, which implies that $e_i$ is stable if $\varrho_i$ is stable. 
We first analyze the stability of $P_{W_{g,i}^*}\varrho_i$:
\begin{align}
    P_{W_{g,i}^*}\dot \varrho_i &= P_{W_{g,i}^*}\left(P_{W_i^*}^\top P_{W_i^*}A_{L_i}x-P_{W_i^*}^\top\bar{A}_{L_i}z_i\right)\notag\\
    &=\Pwgi A_{L_i}x - \Pwgi \Pwi^\top \bar{A}_{L_i}z_i.\label{eq:Pwgi_dot_e_i}
\end{align}
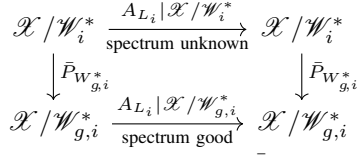
\begin{figure}
    \centering
    \input{Figures/X_W_decompose.tikz}\\[-1.2ex]
    \caption{Commutative diagram illustrating $\bar{\bar A}_{L_i}\bar P_{W_{g\!,i}^*} = \bar P_{W_{g\!,i}^*}\bar A_{L_i}$.}
    \label{fig:X_W_decompose}
\end{figure}
Let 
$
    \bar{\bar A}_{L_i}=\mathrm{Mat}( A_{L_i}|\X/\W_{g,i}^*).
$
According to Fig.~\ref{fig:X_W_decompose}, we have 
$\bar{\bar A}_{L_i}\bar P_{W_{g\!,i}^*} = \bar P_{W_{g\!,i}^*}\bar A_{L_i}$, where the canonical projection $\bar P_{W_{g\!,i}^*}:\X/\W_i^*\rightarrow\X/\W_{g,i}^*$ can be obtained by
$
    \bar P_{W_{g\!,i}^*}=\Pwgi \Pwi^\top.
$
Then, following \eqref{eq:Pwgi_dot_e_i}, we have
\begin{align}
    P_{W_{g,i}^*}\dot \varrho_i &=\Pwgi A_{L_i}x - \Pwgi \Pwi^\top \bar{A}_{L_i}z_i\notag\\
    &=\bar{\bar A}_{L_i}\left(\Pwgi x - \Pwgi \Pwi^\top z_i\right).\label{eq:dot_zeta}
\end{align}
Given $\Pwgi \varrho_i = \Pwgi x - \Pwgi \Pwi^\top z_i$ and \eqref{eq:dot_zeta}, yields
\begin{equation}\label{eq:Pwgie}
    P_{W_{g,i}^*}\dot \varrho_i=\bar{\bar A}_{L_i}P_{W_{g,i}^*} \varrho_i.
\end{equation}
Thanks to $\bar{\bar A}_{L_i}=\mathrm{Mat}( A_{L_i}|\X/\W_{g,i}^*)$ is Hurwitz, $P_{W_{g,i}^*} \varrho_i$ converges to zeros asymptotically. Moreover, since ${W_i^*}^\top \varrho_i=0$ for all $i\in\mathbf{N}_1$, it suffices to prove the stability of $V_i^\top \varrho_i$. Let $\Tilde{\Tilde{A}}_{L_i}:=\mathrm{Mat}(A_{L_i}|\bar\X_{b,i}^*)$. Then, we have
\begin{equation}
\begin{aligned}
    V_i^\top \dot \varrho_i =& \Tilde{\Tilde{A}}_{L_i}V_i^\top\varrho_i-\chi V_i^\top \sum_{j=1}^N a_{ij} (e_i - e_j).
\end{aligned}
\end{equation}
For node $i\in\mathbf{N}_2$, the error $P_{W_{g,i}^*} e_i$ asymptotically converges to zero as well, since its dynamics are identical to \eqref{eq:Pwgie}. We omit this part for space limitations. Hence, for node $i\in\mathbf{N}_2$, it suffices to prove the stability of ${\Wgi}^\top e_i$ as follows
\begin{multline}
       {\Wgi}^\top \dot{e}_i
        ={\Wgi}^\top A_L\left(W_{g,i}^*{W_{g,i}^*}^\top e_i+\Pwgi^\top\Pwgi e_i\right) \\
        + {\Wgi}^\top \bar{B} \bar{u}_i
       -\chi {W_{g,i}^*}^\top \sum_{j=1}^N a_{ij} (e_i - e_j)\\
        -\gamma \mathrm{sign}\left( {W_{g,i}^*}^\top \sum_{j=1}^N a_{ij} (e_i - e_j) \right)
\end{multline}
Let $\Tilde{A}_{L_i}:=\mathrm{Mat}(A_{L_i}|\W_{g,i}^*)$, which can be calculated by $\Tilde{A}_{L_i}={\Wgi}^\top A_{L_i}\Wgi$. As $\Pwgi e_i$ is stable, we only need to investigate the stability of
\begin{align*}
        &{\Wgi}^\top \dot{e}_i = \Tilde{A}_{L_i} {\Wgi}^\top e_i + {\Wgi}^\top \bar{B} \bar{u}_i\\
        &\!-\!\chi {W_{g,i}^*}^\top \sum_{j=1}^N a_{ij} (e_i\! -\! e_j)
        \! -\!\gamma \mathrm{sign}\left(\! {W_{g,i}^*}^\top \sum_{j=1}^N a_{ij} (e_i \! -\! e_j)\! \right)\!.
\end{align*}
Define $
    \bm W=\diag\left(\mathbf{0},\diag_{j\in\mathbf{N}_2}(W_{g,j}^*)\right),\ 
    \bar{\bm B}=\diag\left(\mathbf{0},\diag_{j\in\mathbf{N}_2}(\bar{B}_j)\right),\ 
    \bar{\bm u}=\diag\left(\mathbf{0},\diag_{j\in\mathbf{N}_2}(\bar{u}_j)\right)
$.
Let $e:=\col\left(\col_{i\in\mathbf{N}_1}(e_i),\col_{j\in\mathbf{N}_2}(e_j)\right)$ and $\bar e:=\col\left(\col_{i\in\mathbf{N}_1}(\varrho_i),\col_{j\in\mathbf{N}_2}(e_j)\right)$, then one can obtain
\begin{align}
    \bm W_V^\top\dot{\bar{e}} =& \bm A_L \bm W_V^\top \bar{e} - \chi\bm W_V^\top(\mathcal{L}\otimes I_n)\bar{M}\bar e \notag\\
    &+ \bm W^\top\Bar{\bm B}\bar{\bm u}- \gamma\sign{\left(\bm W^\top(\mathcal{L}\otimes I_n)e\right)}.\label{eq:Wve}
\end{align}
where $\bar M = \diag(\diag_{i\in\mathbf{N}_1}(M_i),I_{nN_2})$.
Consider the following equations 
\begin{align*}
    I_{nN} =& \bm W_V \bm W_V^\top +\diag_{i\in\mathbf{N}}(\Pwgi^\top\Pwgi) \\
    & + \diag\left(\diag_{i\in\mathbf{N}_1}(W_i^*{W_i^*}^\top),\mathbf{0}\right),\\
    I_{nN} =&  \bm W \bm W^\top+ \diag_{i\in\mathbf{N}_2}(\Pwgi^\top\Pwgi)+\Tilde{I},
\end{align*}
with $\Tilde{I}=\diag(I_{\sum_{i\in\mathbf{N}_2}(n-w_{g,i}^*)},\mathbf{0})$. Due to the fact that $\Pwgi \varrho_i$ ($i\in\mathbf{N}_1$) and $\Pwgi e_j$ ($j\in\mathbf{N}_2$) asymptotically converge to zero, ${W_i^*}^\top \varrho_i=0$ for all $i\in\mathbf{N}_1$, and $\bm W^\top(\mathcal{L}\otimes I_n)\Tilde{I}=0$, the convergence of \eqref{eq:Wve} is equivalent to
\begin{equation}
    \begin{aligned}
    &\bm W_V^\top\dot{\bar{e}} = \bm A_L \bm W_V^\top \bar{e} - \chi\bm W_V^\top(\mathcal{L}\otimes I_n)\bar{M}\bm W_V \bm W_V^\top \bar e \\
    &\quad\quad+ \bm W^\top\Bar{\bm B}\bar{\bm u}- \gamma\sign{\left(\bm W^\top(\mathcal{L}\otimes I_n)\bm W \bm W^\top e\right)}.
\end{aligned}\label{eq:Wv_e}
\end{equation}
Let $\varepsilon:=\bm W_V^\top \bar{e}$ and $\Tilde{\varepsilon}:=\bm W^\top e$, \eqref{eq:Wv_e} can be written as
\begin{equation}
    \begin{aligned}
    \dot \varepsilon =& \bm A_L \varepsilon - \chi\bm W_V^\top(\mathcal{L}\otimes I_n)\bar{M}\bm W_V \varepsilon \\
    &+ \bm W^\top\Bar{\bm B}\bar{\bm u}- \gamma\sign{\left(\bm W^\top(\mathcal{L}\otimes I_n)\bm W \Tilde{\varepsilon}\right)}.
\end{aligned}\label{eq:varepsilon}
\end{equation}
To proceed with the convergence of \eqref{eq:varepsilon}, we define the matrix $\bm Q:=\bm W_V^\top(\mathcal{L}\otimes I_n)\bm W_V$, which is positive definite under the Assumption~\ref{asm:connected} and~\ref{asm:geo_con} \cite[Lemma~4]{kim2019completely}. Then, we consider the Lyapunov function $V(\varepsilon) = \varepsilon^\top \bm{Q} \varepsilon$, which has the following time derivative
\begin{multline*}
    \hspace{-2.3mm}\dot V \!=\! 2\varepsilon^\top \!\!\left(\bm{Q}\!\left(\bm A_L\!-\! \chi\bar{\bm{Q}}\right)\right)\varepsilon
    +2\varepsilon^\top\!\!\bm{Q}\bm W^\top\!\Bar{\bm B}\bar{\bm u}\!-\!2\gamma\varepsilon^\top\!\bm{Q}\sign\!{\left(\Tilde{\bm Q} \Tilde{\varepsilon}\right)}\\
     \!=\! 2\varepsilon^\top\! \left(\bm{Q}\left(\bm A_L\! - \!\chi\bar{\bm{Q}}\right)\right)\varepsilon
    +2\Tilde{\varepsilon}^\top\!\Tilde{\bm{Q}}\bm W^\top\!\Bar{\bm B}\bar{\bm u} - 2\gamma\Tilde{\varepsilon}^\top\!\Tilde{\bm{Q}}\sign{\left(\Tilde{\bm Q} \Tilde{\varepsilon}\right)}
\end{multline*}
with $\Tilde{\bm Q}=\bm W^\top(\mathcal{L}\otimes I_n)\bm W$ and $\bar{\bm Q}=\bm{W}_V^\top(\mathcal{L}\otimes I_n)\bar{M}\bm{W}_V$. Recalling the definitions of $\bar M$ and $\bm Q$, and $\|M_i\| \leq 1$, it follows that $\|\bar{\bm Q}\bm{Q}^{-1}\|\leq 1$. Therefore, one can obtain
\begin{align*}
        \dot{V}&\le 2\left(\left\|\bm{A}_L \bm{Q}^{-1}\right\|\!-\!\chi\|\bar{\bm Q}\bm{Q}^{-1}\|\right)\|\bm{Q}\varepsilon\|_2^2 \\
        &\quad + 2\left\|\bar{\bm u}^\top \bar{\bm B}^\top \bm{W}\Tilde{\bm{Q}}\varepsilon \right\|_{\infty} \!\!- \!2 \gamma\|\Tilde{\bm{Q}}\bar\varepsilon\|_1\\
        &\le -2 \left(\chi - \frac{\left\|\bm{A}_L\right\|}{\sigma_{\min}\left(\bm{Q}\right)}\right)\|\bm{Q}\varepsilon\|_2^2 \\
        & - 2 \left( \gamma - \bar{u}_{\max}\max_{i\in\mathbf{N}_2}\left(\left\| \bar{B}_i \right\|_1\right)\max_{i\in\mathbf{N}_2}\left(\left\|W_{g,i}^*\right\|_{\infty}\right)\right)\|\Tilde{\bm{Q}}\bar\varepsilon\|_1.
\end{align*}
The conditions in \eqref{con:gain} ensure that $\dot{V}$ is negative definite. Consequently,  $\varepsilon$ will asymptotically converge to zero, as such $e$ converges to zero, which completes the proof.
\end{proof}
When $\mathbf{N} = \mathbf{N}_1$, i.e., all nodes satisfy the local rank condition \eqref{eq:duio_rank_con}, our geometric condition in {Assumption~\ref{asm:geo_con}} simplifies to $\cap_{i \in \mathbf{N}} \image{V_i} = 0$, which coincides with the {\em Extensive Joint Detectable} condition established in \cite{yang2022state,cao2023distributed,cao2023distributedauto}. Conversely, when $\mathbf{N} = \mathbf{N}_2$, i.e., none of the nodes satisfy the local rank condition \eqref{eq:duio_rank_con}, Assumption~\ref{asm:geo_con} reduces to $\cap_{i \in \mathbf{N}} \W_{g,i}^* = 0$, corresponding to the {\em Extensive Joint Detectable} condition in \cite{zhao2025DUIO}. These results show that our proposed design unifies and extends existing frameworks, encompassing them as special cases.

\section{Simulation Results}\label{sec:simu}
\subsection{Centralized UIO}
Consider the following LTI system
\begin{equation*}
    A\!=\!{\setlength{\arraycolsep}{3pt}\begin{bmatrix}
2 & -2 & 0 \\ 
0 & 0 & 1 \\
0 & -2 & 1
\end{bmatrix}}\!,\ B\!=\!\begin{bmatrix}
    \acute B&\bar{B}
\end{bmatrix}=\left[\begin{array}{@{\hskip 0.1em}c;{2pt/2pt}c@{\hskip 0.1em}}
0 & 1 \\
0 & 1 \\
1 & 0
\end{array}\right]\!,\ C\!=\!{\setlength{\arraycolsep}{3pt}\begin{bmatrix}
1 & 0 & 0 \\
0 & 1 & 0
\end{bmatrix}},\ 
\end{equation*}
with initial state $x_0=\matrices{1&2&3}^\top$ and inputs $u=\matrices{\acute u& \bar{u}}^\top$, where $\acute u=\mathtt{sin}(t)$ and $\bar{u}=\mathtt{cos}(0.5t)$. Following the centralized UIO design in Section~\ref{sec:centralized_UIO}, we can calculate the observer gains and matrices, which are provided in the supplementary document\footnote{\label{docm:parameter}\href{https://github.com/RuixuanZhaoEEEUCL/CDC2025.git}{https://github.com/RuixuanZhaoEEEUCL/CDC2025.git}} due to space limitations. 
Fig.~\ref{fig:estimation_errors_cen_uio} shows the simulation results, where the system states are reconstructed despite the presence of unknown inputs.
\begin{figure}
    \centering
    \begin{subfigure}[b]{0.24\textwidth}
        \centering
        \includegraphics[width=\textwidth]{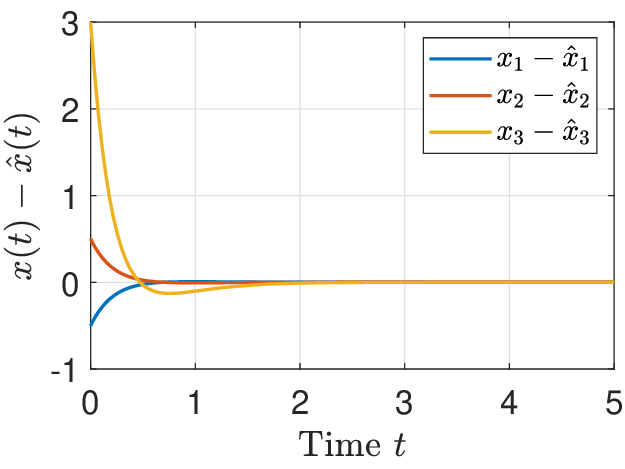}
        \\[-1.6ex]
        \caption{}
        \label{fig:estimation_errors_cen_uio}
    \end{subfigure}\hspace{-4pt}
    \begin{subfigure}[b]{0.24\textwidth}
        \centering
        \includegraphics[width=\textwidth]{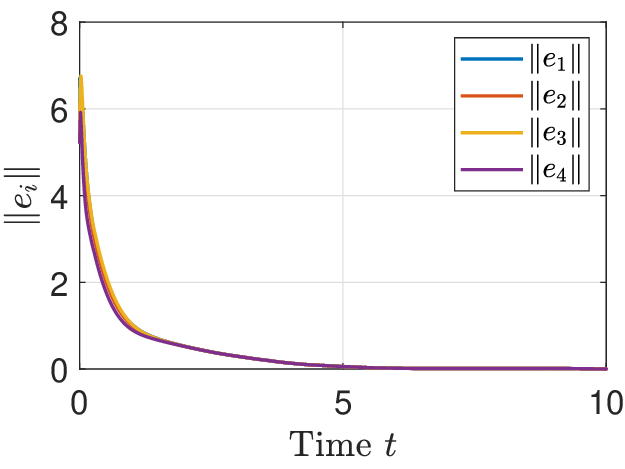}
        \\[-1.6ex]
        \caption{}
        \label{fig:estimation_errors_duio}
    \end{subfigure}
    \\[-1.2ex]
    \caption{(a) Estimation errors of the centralized UIO designed in \eqref{eq:cen_uio}. (b) Estimation errors of the distributed UIO designed by \eqref{eq:DUIO_1} and \eqref{eq:DUIO_2}.}
\end{figure}
\subsection{Distributed UIO}
Consider an LTI system measured by 4 sensor nodes, each node $i$ has access only to its local input $u_i$, local measurement $y_i$, and shared information from neighboring nodes via a communication network. Fig.~\ref{fig:network} illustrates the architecture of the distributed UIO with communication topology. System parameters are given as follows
\begin{align*}
 &A\!=\!{\setlength{\arraycolsep}{1.45pt}\begin{bmatrix}
0 & 3 & 0 & 0 & 0 & 0 \\
-2 & 0 & 1 & 0 & 0 & 0 \\
0 & 0 & 0 & 2 & 0 & 0 \\
0 & 0 & -3 & -2 & 0 & 0 \\ 
0 & 0 & 0 & 1 & 0 & -3 \\
0 & 2 & 0 & 0 & 4 & 0
\end{bmatrix}}\!,\ 
B^\top \!=\!\left[\scriptstyle\begin{array}{@{\hskip 0.1em}c@{\hskip 0.5em}c@{\hskip 0.5em}c@{\hskip 0.5em}c@{\hskip 0.5em}c@{\hskip 0.5em}c@{\hskip 0.1em}}
0 & 1 & 0 & 0 & 0 & 1 \\ \hdashline
0 & 0 & 0 & 1 & 0 & 0 \\ \hdashline
0 & 0 & 1 & 0 & 0 & 1
\end{array}\right]\!=\!\begin{bmatrix}
        \bm b_1^\top\\ \bm b_2^\top\\ \bm b_3^\top
    \end{bmatrix}\!,\\
    &C_1 = {\setlength{\arraycolsep}{4pt}\begin{bmatrix}
            1 & 0 & 0 & 0 & 0 & 0 \\
            0 & 0 & 1 & 0 & 0 & 0
    \end{bmatrix}},\ 
    C_3 = {\setlength{\arraycolsep}{4pt}\begin{bmatrix}
            0 & 0 & 1 & 0 & 0 & 0 \\
            0 & 1 & 0 & 0 & 0 & 0
    \end{bmatrix}},\\
    &C_2={\setlength{\arraycolsep}{4pt}\begin{bmatrix}
        0 & 1 & 0 & 0 & 1 & 0
    \end{bmatrix}},\quad
    C_4 = {\setlength{\arraycolsep}{4pt}\begin{bmatrix}
        1 & 1 & 0 & 0 & 0 & 0
    \end{bmatrix}},
\end{align*}
with initial state $x_0=\matrices{1 & 2 & 3 & -1 & -2 & -3}^\top$
and inputs $u(t)=\matrices{\mathtt{sin}(t)&0.2\mathtt{cos}(t)&0.2\mathtt{sin}(0.5t)}^\top$. The known and unknown input channels for each sensor node are given by 
$
    B_1 = \matrices{\bm b_1&\bm b_2},\ \bar B_1=\bm b_3,\ 
    B_2 = \matrices{\bm b_1&\bm b_3},\ \bar B_2=\bm b_2,\,
    B_3 = \matrices{\bm b_2&\bm b_3}=\bar{B}_4,\ \bar B_3=\bm b_1=B_4.
$
and $u_i$, $\bar u_i$ can be characterized accordingly.
\begin{figure}[htp]
    \centering
    \scalebox{0.85}{\input{Figures/observers_diagram.tikz}}
    \\[-0.5ex]
    \caption{A distributed UIO example is considered, consisting of 4 sensor nodes, where $\mathbf{N}_1 = \{1,3\}$ and $\mathbf{N}_2 = \{2,4\}$. Each node communicates its local state estimates to neighboring nodes via a communication network represented by the cyan dashed lines.
    }
    \label{fig:network}
\end{figure}
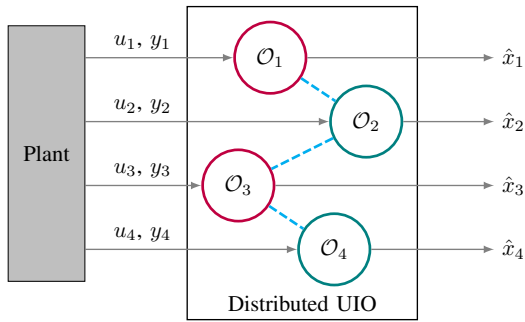
According to the distributed UIO design in Section~\ref{sec:distributed_uio}, we can calculate observer gains and matrices, which are provided in the supplementary document\textsuperscript{\ref{docm:parameter}}. Fig.~\ref{fig:estimation_errors_duio} illustrates the asymptotic convergence of the estimation errors at each node, which validates the effectiveness of our design.

\section{Conclusion}\label{sec:conclusion}
In this paper, we propose a novel design methodology for unknown input observers (UIOs) tailored to linear time-invariant (LTI) systems, applicable to both centralized and distributed architectures. The proposed approach is grounded in the geometric approach through which we derive generalized conditions that extend beyond those found in the existing literature, thereby improving applicability, especially for distributed state estimation problems. The simulation results demonstrate the effectiveness of the proposed designs. Future work will focus on extending the framework to nonlinear systems and incorporating both uncertainties.


\end{document}

%% file: Figures/Wg_decompose.tikz
\begin{tikzpicture}


\def\xgap{3.25}
\def\ygap{1.5}
\foreach \i in {0,...,3}{
    \foreach \j in {0,...,3}{
        \coordinate (\i\j) at (\i*\xgap, \j*\ygap);
    }
}
\pgfmathatantwo{\ygap}{\xgap}
\def\labangle{\pgfmathresult}

\tikzset{label/.style = {font=\small, midway}}

\node (barU) at (01) {$\bar \U$};
\node (XmodWg) at (10) {$\X / \W^*_{g}$};
\node (XmodWg2) at (20) {$\X / \W^*_{g}$};
\node (X) at (11) {$\X$};
\node (X2) at (21) {$\X$};
\node (Wg) at (12) {$\W^*_{g}$};
\node (Wg2) at (22) {$\W^*_{g}$};

\draw[-latex, dashed] (barU) -- (XmodWg) node [label, below, rotate=-\labangle] {$0$};
\draw[-latex] (barU) -- (X) node [label, above] {$\bar B$};

\draw[-latex] (XmodWg) -- (XmodWg2) node [label, above] {$A_{L}|\X/\W_{g}^*$};
\draw[-latex] (XmodWg) -- (XmodWg2) node [label, below] {spectrum good};
\draw[-latex] (X) -- (X2) node [label, above] {$A_{L}$};
\draw[-latex] (Wg) -- (Wg2) node [label, above] {$A_{L}|\W_{g}^*$};

\draw[-latex] (Wg) -- (X) node [label, right] {$W_{g}^*$};
\draw[-latex] (Wg2) -- (X2) node [label, right] {$W_{g}^*$};

\draw[-latex] (X) -- (XmodWg) node [label, right] {$P_{W_{g}^*}$};
\draw[-latex] (X2) -- (XmodWg2) node [label, right] {$P_{W_{g}^*}$};

\end{tikzpicture}

%% file: Figures/forremark.tikz
\begin{tikzpicture}


\def\xgap{1.5}
\def\ygap{1.5}
\foreach \i in {0,...,5}{
    \foreach \j in {0,...,5}{
        \coordinate (\j\i) at (\i*\xgap, \j*\ygap);
    }
}
\pgfmathatantwo{\ygap}{\xgap}
\def\labangle{\pgfmathresult}

\tikzset{label/.style = {font=\small, midway}}


\node (U) at (10) {$\bar{\mathscr U}$};
\node (XmodS) at (01) {$\mathscr X / \image \bar B$};
\node (X) at (11) {$\mathscr X$};
\node (X2) at (13) {$\mathscr X$};
\node (Y) at (15) {$\mathscr Y$};
\node (XmodS2) at (03) {$\mathscr X / \image \bar B$};
\node (YmodZ) at (05) {$\Y/(C\image \bar B)$};

\draw[-latex] (U) -- (X) node [label, above] {$\bar B$};
\draw[-latex] (X) -- (XmodS) node [label, right] {$P$};
\draw[-latex, dashed] (U) -- (XmodS) node [label, below, rotate=-\labangle] {$0$};



\draw[-latex] (X) -- (X2) node [label, above] {$A_L$};
\draw[-latex] (X2) -- (Y) node [label, above] {$C$};

\draw[-latex] (X2) -- (XmodS2) node [label, right] {$P$};
\draw[-latex] (X2) -- (YmodZ) node [label, above, rotate=-\labangle] {$\bar C_Y$};
\draw[-latex] (Y) -- (YmodZ) node [label, right] {$P_Y$};

\draw[-latex] (XmodS) -- (XmodS2) node [label, above] {$\bar A_L$};
\draw[-latex] (XmodS2) -- (YmodZ) node [label, above] {$\bar{C}$};
\end{tikzpicture}

%% file: Figures/X_W_decompose.tikz
\begin{tikzcd}[column sep=50pt]
\X/\W_i^* \arrow{d}{\bar P_{W_{g\!,i}^*}} \arrow{r}{ A_{L_i}|\X/\W_i^*}[swap]{\text{spectrum unknown}} & \X/\W_i^* \arrow{d}{\bar P_{W_{g\!,i}^*}}  \\
\X/\W_{g,i}^* \arrow{r}{ A_{L_i}|\X/\W_{g,i}^*}[swap]{\text{spectrum good}} & \X/\W_{g,i}^* 
\end{tikzcd}

%% file: Figures/observers_diagram.tikz
\begin{tikzpicture}
\def\off{18}
\def\N{5}
\def\R{2.5}
\pgfmathparse{360/\N}
\edef\step{\pgfmathresult}

\colorlet{net}{teal}
\tikzset{comm/.style = {color=cyan, very thick, dash pattern=on 4pt off 1.5pt}}

\node [circle, 
        draw, 
        color=purple, 
        fill=white, 
        text=black, 
        very thick,
        inner sep=6pt] (O1) at (2.5,1.5) {$\mathcal O_{1}$};
\node [circle, 
        draw, 
        color=net, 
        fill=white, 
        text=black, 
        very thick,
        inner sep=6pt] (O2) at (4,0.5) {$\mathcal O_{2}$};
\node [circle, 
        draw, 
        color=purple, 
        fill=white, 
        text=black, 
        very thick,
        inner sep=6pt] (O3) at (2, -0.5) {$\mathcal O_{3}$};
\node [circle, 
        draw, 
        color=net, 
        fill=white, 
        text=black, 
        very thick,
        inner sep=6pt] (O4) at (3.5, -1.5) {$\mathcal O_{4}$};
\draw [comm] (O1) -- (O2);
\draw [comm] (O2) -- (O3);
\draw [comm] (O3) -- (O4);
\draw[semithick, color=black] (1.2,-2.6) -- (4.8,-2.6) -- (4.8,2.3) -- (1.2,2.3) -- cycle;

\node[above] at (3,-2.6) {Distributed UIO};

\node [draw,
        rectangle,
        fill = lightgray,
        minimum width=1.2cm,
        minimum height=4cm] (sys) at (-1,0) {Plant}; 

\draw [-latex, semithick, color=gray] (-0.4,1.5) -- (O1) node[midway, above, align=left] {};
\draw [-latex, semithick, color=gray] (-0.4,0.5) -- (O2) node[midway, above, align=left] {};
\draw [-latex, semithick, color=gray] (-0.4,-0.5) -- (O3) node[midway, above, align=left] {};
\draw [-latex, semithick, color=gray] (-0.4,-1.5) -- (O4) node[midway, above, align=left] {};
\node[above] at (0.5,1.5) {$u_1$,\ $y_1$};
\node[above] at (0.5,0.5) {$u_2$,\ $y_2$};
\node[above] at (0.5,-0.5) {$u_3$,\ $y_3$};
\node[above] at (0.5,-1.5) {$u_4$,\ $y_4$};

\draw [-latex, semithick, color=gray] (O1) -- (6,1.5) node[midway, above, align=left] {};
\draw [-latex, semithick, color=gray] (O2) -- (6,0.5) node[midway, above, align=left] {};
\draw [-latex, semithick, color=gray] (O3) -- (6,-0.5) node[midway, above, align=left] {};
\draw [-latex, semithick, color=gray] (O4) -- (6,-1.5) node[midway, above, align=left] {};
\node[right] at (6,1.5) {$\hat{x}_1$};
\node[right] at (6,0.5) {$\hat{x}_2$};
\node[right] at (6,-0.5) {$\hat{x}_3$};
\node[right] at (6,-1.5) {$\hat{x}_4$};
\end{tikzpicture}

%% file: main.bbl
\begin{thebibliography}{10}
\providecommand{\url}[1]{#1}
\csname url@samestyle\endcsname
\providecommand{\newblock}{\relax}
\providecommand{\bibinfo}[2]{#2}
\providecommand{\BIBentrySTDinterwordspacing}{\spaceskip=0pt\relax}
\providecommand{\BIBentryALTinterwordstretchfactor}{4}
\providecommand{\BIBentryALTinterwordspacing}{\spaceskip=\fontdimen2\font plus
\BIBentryALTinterwordstretchfactor\fontdimen3\font minus \fontdimen4\font\relax}
\providecommand{\BIBforeignlanguage}[2]{{%
\expandafter\ifx\csname l@#1\endcsname\relax
\typeout{** WARNING: IEEEtran.bst: No hyphenation pattern has been}%
\typeout{** loaded for the language `#1'. Using the pattern for}%
\typeout{** the default language instead.}%
\else
\language=\csname l@#1\endcsname
\fi
#2}}
\providecommand{\BIBdecl}{\relax}
\BIBdecl

\bibitem{chen1996design}
J.~Chen, R.~J. Patton, and H.-Y. Zhang, ``Design of unknown input observers and robust fault detection filters,'' \emph{International Journal of control}, vol.~63, no.~1, pp. 85--105, 1996.

\bibitem{darouach1994full}
M.~Darouach, M.~Zasadzinski, and S.~J. Xu, ``Full-order observers for linear systems with unknown inputs,'' \emph{IEEE Transactions on Automatic Control}, vol.~39, no.~3, pp. 606--609, 1994.

\bibitem{valcher1999state}
M.~E. Valcher, ``State observers for discrete-time linear systems with unknown inputs,'' \emph{IEEE Transactions on Automatic Control}, vol.~44, no.~2, pp. 397--401, 1999.

\bibitem{bejarano2009unknown}
F.~J. Bejarano, L.~Fridman, and A.~Poznyak, ``Unknown input and state estimation for unobservable systems,'' \emph{SIAM Journal on Control and Optimization}, vol.~48, no.~2, pp. 1155--1178, 2009.

\bibitem{trentelman2002control}
H.~L. Trentelman, A.~A. Stoorvogel, M.~Hautus, and L.~Dewell, ``Control theory for linear systems,'' \emph{Appl. Mech. Rev.}, vol.~55, no.~5, pp. B87--B87, 2002.

\bibitem{turan2021data}
M.~S. Turan and G.~Ferrari-Trecate, ``Data-driven unknown-input observers and state estimation,'' \emph{IEEE Control Systems Letters}, vol.~6, pp. 1424--1429, 2021.

\bibitem{disaro2024data}
G.~Disar{\`o} and M.~E. Valcher, ``Data-driven reduced-order unknown-input observers,'' \emph{European Journal of Control}, vol.~80, p. 101034, 2024.

\bibitem{disaro2024equivalence}
------, ``On the equivalence of model-based and data-driven approaches to the design of unknown-input observers,'' \emph{IEEE Transactions on Automatic Control}, 2024.

\bibitem{kim2019completely}
T.~Kim, C.~Lee, and H.~Shim, ``Completely decentralized design of distributed observer for linear systems,'' \emph{IEEE Transactions on Automatic Control}, vol.~65, no.~11, pp. 4664--4678, 2019.

\bibitem{han2018simple}
W.~Han, H.~L. Trentelman, Z.~Wang, and Y.~Shen, ``A simple approach to distributed observer design for linear systems,'' \emph{IEEE Transactions on Automatic Control}, vol.~64, no.~1, pp. 329--336, 2018.

\bibitem{wang2024split}
L.~Wang, J.~Liu, B.~D. Anderson, and A.~S. Morse, ``Split-spectrum based distributed state estimation for linear systems,'' \emph{Automatica}, vol. 161, p. 111421, 2024.

\bibitem{yang2022sensor}
G.~Yang, H.~Rezaee, A.~Serrani, and T.~Parisini, ``Sensor fault-tolerant state estimation by networks of distributed observers,'' \emph{IEEE Transactions on Automatic Control}, vol.~67, no.~10, pp. 5348--5360, 2022.

\bibitem{yang2023plug}
G.~Yang, A.~Barboni, H.~Rezaee, A.~Serrani, and T.~Parisini, ``Plug-and-play design for linear distributed observers,'' \emph{IFAC-PapersOnLine}, vol.~56, no.~2, pp. 10\,811--10\,816, 2023.

\bibitem{zhao2024state}
R.~Zhao, G.~Yang, P.~Li, T.~Parisini, and B.~Chen, ``State estimation using a network of observers for a class of nonlinear systems with communication delay,'' in \emph{2024 European Control Conference (ECC)}.\hskip 1em plus 0.5em minus 0.4em\relax IEEE, 2024, pp. 762--767.

\bibitem{ge2022fixed}
P.~Ge, P.~Li, B.~Chen, and F.~Teng, ``Fixed-time convergent distributed observer design of linear systems: A kernel-based approach,'' \emph{IEEE Transactions on Automatic Control}, vol.~68, no.~8, pp. 4932--4939, 2022.

\bibitem{yang2022state}
G.~Yang, A.~Barboni, H.~Rezaee, and T.~Parisini, ``State estimation using a network of distributed observers with unknown inputs,'' \emph{Automatica}, vol. 146, p. 110631, 2022.

\bibitem{cao2023distributed}
G.~Cao and J.~Wang, ``Distributed unknown input observer,'' \emph{IEEE Transactions on Automatic Control}, 2023.

\bibitem{cao2023distributedauto}
------, ``A distributed reduced-order unknown input observer,'' \emph{Automatica}, vol. 155, p. 111174, 2023.

\bibitem{zhao2025DUIO}
R.~Zhao, G.~Yang, T.~Parisini, and B.~Chen, ``Distributed unknown input observer design with relaxed conditions: Theory and application to vehicle platooning,'' in \emph{2025 European Control Conference (ECC)}.\hskip 1em plus 0.5em minus 0.4em\relax IEEE, 2025, pp. 1408--1413.

\bibitem{wonham1985linear}
W.~M. Wonham, \emph{Linear Multivariable Control a Geometric Approach}, 3rd~ed.\hskip 1em plus 0.5em minus 0.4em\relax New York, NY, USA: Springer, 1985.

\bibitem{massoumnia1986geometric}
M.-A. Massoumnia, ``A geometric approach to failure detection and identification in linear systems,'' Ph.D. dissertation, Massachusetts Institute of Technology, 1986.

\end{thebibliography}
